\documentclass[twocolumn]{article}

\usepackage[english]{babel}
\usepackage[utf8x]{inputenc}
\usepackage[T1]{fontenc}

\usepackage[a4paper,top=3cm,bottom=2cm,left=3cm,right=3cm,marginparwidth=1.75cm]{geometry}

\usepackage{natbib}

\usepackage{amsmath}
\usepackage{amssymb}
\usepackage{amsthm}
\usepackage{graphicx}
\usepackage[colorinlistoftodos]{todonotes}
\usepackage[colorlinks=true, allcolors=blue]{hyperref}

\usepackage{color, colortbl}
\definecolor{Gray}{gray}{0.9}
\definecolor{LightCyan}{rgb}{0.88,1,1}

\newcommand{\Expect}{{\rm I\kern-.3em E}}

\newtheorem{theorem}{Theorem}[section]   

\newtheorem{proposition}[theorem]{Proposition}  

\newtheorem{definition}[theorem]{Definition}   

\newtheorem{remark}[theorem]{Remark}        
\newtheorem{example}[theorem]{Example}        

\title{Quantifying the Total Effect of Edge Interventions in Discrete Multistate Networks}
\author{David Murrugarra$^{a}$ \and Elena Dimitrova$^{b}$}

\begin{document}
\maketitle

{\footnotesize
     \centerline{$^a$Department of Mathematics,
      University of Kentucky,}
  \centerline{Lexington, KY 40506-0027, USA}
}
{\footnotesize
     \centerline{$^b$Department of Mathematics,
      California Polytechnic}
  \centerline{State University, San Luis Obispo, CA 93407-0403, USA}
}
\begin{abstract}
Developing efficient computational methods to assess the impact of external interventions on the dynamics of a network model is an important problem in systems biology. This paper focuses on quantifying the global changes that result from the application of an intervention to produce a desired effect, which we define as the \emph{total effect} of the intervention. The type of mathematical models that we will consider are discrete dynamical systems which include the widely used Boolean networks and their generalizations. The potential interventions can be represented by a set of nodes and edges that can be manipulated to produce a desired effect on the system. We use a class of regulatory rules called nested canalizing functions that frequently appear in published models and were inspired by the concept of canalization in evolutionary biology. In this paper, we provide a polynomial normal form based on the canalizing properties of regulatory functions. Using this polynomial normal form, we give a set of formulas for counting the maximum number of transitions that will change in the state space upon an edge deletion in the wiring diagram. These formulas rely on the canalizing structure of the target function since the number of changed transitions depends on the canalizing layer that includes the input to be deleted. We also present computations on random networks to compare the exact number of changes with the upper bounds provided by our formulas. Finally, we provide statistics on the sharpness of these upper bounds in random networks.
\end{abstract}
\section{Introduction}
Boolean networks (BN) have been proposed as an appropriate framework for modeling the state of cells due to their simplicity and the variety of tools available for model analysis~\cite{Veliz-Cuba:2010aa,Zanudo:2015aa}. However, some complex gene interactions cannot be represented in the Boolean setting and several generalizations of the Boolean approach have been developed~\cite{Thomas:1990aa}. Multistate models, a generalization of the BN framework, where the genes can attain more than two states have been proposed as appropriate models for capturing complex gene expression patterns, such as consideration of three states (low, medium, and high). We note that while in theory it is possible to develop models where the variables can take on any number of possible states (possibly only restricted by the requirement that it needs to be a power of a prime number so that the domain can have the structure of a finite field), in practical applications this number is typically small -- rarely larger than 5.

A Gene Regulatory Network (GRN) is a representation of the intricate relationships among genes, proteins, and other compounds that are responsible for the expression levels of mRNA and proteins.
Boolean networks have been successfully used to model and study the properties of GRN~\cite{Albert:2003aa,Li:2004aa}. In particular, Boolean canalizing rules were introduced by S.~Kauffman and collaborators~\cite{Kauffman2003,Kauffman:2004aa} and reflect the concept of canalization in evolutionary biology that Waddington pioneered in 1942~\cite{Waddington:1957aa} -- that organisms evolve developmental robustness, producing an invariant phenotype even under genetic or environmental perturbations.

In this article, we study the network-wide effect of an experimental intervention that either prevents a regulation from happening or silences a node. Such intervention is modeled through edge deletion and can be achieved via therapeutic drugs that target a specific gene interaction~\cite{Choi2012,Campbell:2019aa}. In~\cite{Murrugarra:2015uq} we introduced methods for quantifying side effects in Boolean networks. However, many of the more recently published discrete dynamical models include variables that take on more than two states due to the need for capturing mechanisms that are not binary in nature~\cite{zanudo2017network,Remy_2015,Chifman:2017aa,Espinosa-Soto:2004aa,Thieffry:1995aa}. Consequently, Boolean nested and partially nested canalizing functions were generalized to multistate~\cite{Murrugarra:2011aa,Murrugarra:2012aa,kadelka2017multistate} which enables the possibility of capturing more complex interactions among the genes in the network. Such functions can be viewed as a discrete dynamical system with a stratified structure which consists of hierarchical layers of variables according to their relative influence over the system dynamics.

The ability to quantify the global changes in the dynamics of the network after an external perturbation has important applications. In the presence of external network modifications where the topology of the network changes but the attractor structure remains unchanged, it is still desirable to quantify the changes in other aspects of the dynamics such as the transient time. For instance,  in evolutionary biology to simulate evolution one often evolves an ensemble of networks (by performing mutations, crossover, selection, etc.) for many generations~\cite{Wagner1996}. At the end of the simulations, one measures the changes in the evolved networks to compare with the features of the original ensemble. In~\cite{Wagner1996}, after simulated evolution, the evolved networks had similar features to the original ones such as the number of attractors and basin sizes. One feature that had changed is the transient time~\cite{Wagner1996}. The theoretical tools presented in this paper will be useful to measure global changes even if the attractor structure is preserved after an intervention.

There are several published control methods for Boolean networks such as Stable Motifs~\cite{Zanudo:2015aa}, Feedback Vertex Sets~\cite{Zanudo:2017wq}, Minimal Hitting Sets~\cite{Vera-Licona:2013aa,Klamt:2006mi}, and several others~\cite{Qiu:2014aa,Li:2015aa,PMID:25433558,gates2016control,zanudo2017network,Sordo-Vieira:2019aa}. 
While these control methods focus on finding control
interventions that satisfy a control objective (e.g., to drive the
system into a specific attractor), there are very few studies of the
total extent of the consequences of applying a certain control action
(beyond satisfying the control objective).
This paper contributes methods for measuring the impact of the control actions on the dynamics of multistate networks. The type of theoretical tools presented here can help to discriminate control actions with minimal effect on the state space. That is, even if we have different control candidates that can achieve a certain objective, they might have different impact on the dynamics of the network and we might be interested in distinguishing the control action that produces the least changes in the dynamics of the network.

The rest of the paper is structured as follows. In Section~\ref{sec:Background}, we introduce discrete dynamical systems and their representation as polynomial dynamical systems. In Section~\ref{sec:Methods} we define the control actions for multistate networks. In Section~\ref{sec:Results} we provide a polynomial normal form for discrete functions and then we use this representation to derive a set of formulas for counting the maximum number of transitions in the state space upon edge deletions. In Section~\ref{sec:Applications} we present computational results for random networks. Finally, in Section~\ref{sec:Conlusions}, we provide the conclusions of the paper.
\section{Background}\label{sec:Background}
A discrete dynamical system can be defined as a dynamical system that is discrete in time as well as 
in variable states. More formally,
consider a collection $x_1, \ldots , x_n$ of
variables, each of which can take on values in finite sets $X_1,\dots,X_n$.
Let $X = X_1\times\dots\times X_n$ be their Cartesian product.
A discrete dynamical system in the variables $x_1, \ldots , x_n$ is a function 
\begin{displaymath}
\mathbf{ F} = (f_1,\dots,f_n):X\rightarrow X
\end{displaymath}
where each coordinate function $f_i$ is a discrete function on a subset of $\{x_1,\dots,x_n\}$ which represents how the future value of the $i$-th variable depends on the present values of the variables. If $X_i=\{0,1\}$, then each $f_i$ is a Boolean rule and $\mathbf{ F}$ is a Boolean network. 

In this article, for the purpose of exploiting the algebraic properties of discrete functions, it is assumed that the variables $x_1, \ldots , x_n$ take on values from a finite field $\mathbb{F}$. Then using the fact that any discrete function $f_i:\mathbb{F}^n\rightarrow \mathbb{F}$ can be represented as a polynomial in $x_1,\dots,x_n$, that is $f_i\in\mathbb{F}[x_1,\dots,x_n]$,
the discrete network can be represented as $\mathbf{ F} = (f_1,\dots,f_n):\mathbb{F}^n\rightarrow \mathbb{F}^n$
where each $f_i\in\mathbb{F}[x_1,\dots,x_n]$. If any of the variables $x_1, \ldots , x_n$ take on values from a set that cannot be directly identified with a finite field, then it is straightforward to embed the system $\mathbf{F}:X\rightarrow X$ into a system $\hat{\mathbf{F}}:\mathbb{F}^n\rightarrow \mathbb{F}^n$, where $X\subset \mathbb{F}^n$, while preserving the attractor structure of $\mathbf{F}$; see~\cite{Veliz-Cuba:2010aa}.

Given a discrete network $\mathbf{ F} = (f_1,\dots,f_n)$, a directed graph $\mathcal{W}$ on $n$ nodes 
$x_1, \ldots , x_n$ is associated to $\mathbf{ F}$ as follows: there is a directed edge in $\mathcal{W}$ from 
$x_j$ to $x_i$ if $x_j$ appears in $f_i$, i.e. $x_j$ is in the \emph{support} of $f_i$, written $x_j\in supp(f_i)$. In the context
 of a molecular network model, this graph represents the wiring diagram of the network.

The dynamics of a discrete network is given by the difference equation $x(t+1)=\mathbf{ F}(x(t))$; that is, the dynamics is generated by iteration of $\mathbf{ F}$. More precisely, the dynamics of $\mathbf{ F}$ is represented by the state space graph $S$, defined as the graph with vertices in
$\mathbb{F}^n$ which has an edge from $x\in \mathbb{F}^n$ to $y\in \mathbb{F}^n$ if and only 
if $y = \mathbf{ F}(x)$. In this context, the problem of finding the states
$x\in \mathbb{F}^n$ where the system will get stabilized is of particular importance. The collection of these special points of the state space are called \emph{attractors} of a discrete network and elements of the attractors may include steady states (fixed points), where $\mathbf{ F}(x) = x$, or cycles, where $\mathbf{ F}^r(x) = x$ for some integer $r>1$. Attractors in network modeling might represent cell types~\cite{Kauffman:1969aa} or cellular states such as apoptosis, proliferation, or cell senescence~\cite{Huang:1999aa,DBLP:books/daglib/0024105}.
\section{Methods}\label{sec:Methods}
Network interventions can be modeled through edge and node manipulations and can be achieved via therapeutic drugs that target a specific gene interaction~\cite{Choi2012,Campbell:2019aa}. 
In~\cite{Murrugarra:2015uq,Murrugarra:2016aa} we provided definitions for these actions in Boolean networks. These definitions are usually used for encoding the control parameters with the purpose of identifying control targets as shown in~\cite{Murrugarra:2016aa}. In this paper we will consider the deletion and constant expression of edges in the multistate setting.

Throughout the paper, $S_{i,j}$ will be a subset of $\mathbb{F}$ and $Q_{i,j}(u)$ will be the indicator functions of $S_{i,j}$. 
That is, they return 1 when $u$ is in the set and 0 when $u$ is not. The index $i$ indicates the node $x_i$ from which the edge begins and the second index $j$ is used when necessary to mark the function under consideration.
\subsection{Edge Control in Multistate Networks}
 
In the Boolean setting, the deletion of an edge was implemented by setting an input to zero so that the interaction of that input (represented by an edge) was being silenced. For the multistate case, the silencing of the interaction will be applied whenever the control variable is within a range of values of the possible discrete values.

\begin{definition}[Edge Deletion]\label{def:edge_del_m}
Consider the edge $x_i\rightarrow x_j$ in a wiring diagram. 
For $u\in S_{i,j}$, the control of the edge $x_i\rightarrow x_j$ consists of manipulating the input variable $x_i$ for $f_j$ in the following way: 
\begin{equation*}
\label{edge_del_def}
\mathcal{F}_j(x,u) = f_j(x_{j_1},\dots,(1-Q_{i,j}(u))x_i,\dots,x_{j_m}).
\end{equation*}
For each value of $u\in\mathbb{F}$ we have the following control settings: 
\begin{itemize}
  \item For $u\in S_{i,j}$, 
  $$\mathcal{F}_j(x,u) = f_j(x_{j_1},\dots,x_i=0,\dots,x_{j_m}).$$
  That is, the control is active and the action represents the removal of the edge $x_i\rightarrow x_j$.
  \item For $u\notin S_{i,j}$, 
  $$\mathcal{F}_j(x,u) = f_j(x_{j_1},\dots,x_i,\dots,x_{j_m}).$$ 
  That is, the control is not active.
\end{itemize}
\end{definition}
We will also consider the constant expression of edges, which we define as follows.
\begin{definition}[Constant expression]\label{def:edge_ce_m}
Consider the edge $x_i\rightarrow x_j$ in a wiring diagram and $a\in\mathbb{F}$. 
For $u\in S_{i,j}$, the control of the edge $x_i\rightarrow x_j$ consists of manipulating the input variable $x_i$ for $f_j$ in the following way: 
\begin{equation}
\begin{array}{l}
\mathcal{F}_j(x,u) = \\
f_j(x_{j_1},\dots,(1-Q_{i,j}(u))x_i+aQ_{i,j}(u),\dots,x_{j_m})\nonumber.
\end{array}
\end{equation}
For each value of $u\in\mathbb{F}$ we have the following settings: 
\begin{itemize}
  \item For $u\in S_{i,j}$, 
  $$\mathcal{F}_j(x,u) = f_j(x_{j_1},\dots,x_i=a,\dots,x_{j_m}).$$ 
  That is, the control is active and the action represents the constant expression (to $a$) of the edge $x_i\rightarrow x_j$.
  \item For $u\notin S_{i,j}$, 
  $$\mathcal{F}_j(x,u) = f_j(x_{j_1},\dots,x_i,\dots,x_{j_m}).$$ 
  That is, the control is not active.
\end{itemize}
\end{definition}
\begin{remark}
Node control of $x_i$ can be defined as setting the function $f_i$ to a constant $a\in\mathbb{F}$.
\end{remark}
\section{Results}\label{sec:Results}
In this section we present a definition of $k$-canalizing functions for the multistate case and then we characterize this functions in terms of layers of canalizations. Subsequently, we use this canalizing layers representation to derive an upper bound for the number of changes in the state space of a discrete system upon an edge deletion in the wiring diagram.
\subsection{Multistate $k$-Canalizing Functions}
In the following definition, we assume that $\sigma$ is a permutation on $\{1,\dots,n\}$. 
\begin{definition}\label{Defn_k_can}
The function $f:\mathbb{F}^n\rightarrow \mathbb{F}$ is a $k$-\emph{canalizing function} in the variable order $x_{\sigma(1)},\dots,x_{\sigma(k)}$ with \emph{canalizing input sets} $S_1,\dots,S_k\subset \mathbb{F}$ and \emph{canalizing output values} $b_1,\dots,b_{k}\in\mathbb{F}$ if it can be represented in the form
\begin{equation}\label{eq:defn_k_can}
\begin{array}{l}
f(x_1,\dots,x_n)=\\ \\
\left\{
\begin{array}{l}
b_1,\text{ if}\ x_{\sigma(1)}\in S_1,\\
b_2,\text{ if}\ x_{\sigma(1)}\notin S_1,x_{\sigma(2)}\in S_2,\\
\vdots\\
b_k,\text{ if}\ x_{\sigma(1)}\notin S_1,\dots,x_{\sigma(k)}\in S_k,\\
g\neq b_{k},\text{ if}\ x_{\sigma(1)}\notin S_1,\dots,x_{\sigma(k)}\notin S_k,
\end{array}\right.
\end{array}
\end{equation}
where $g=g(x_{\sigma(k+1)},\dots,x_{\sigma(n)})$ is a multistate function on $n-k$ variables. When $g$ is not a canalizing function, the integer $k$ is the canalizing \emph{depth} of $f$. If $g$ is not a constant function, then $g$ is called the \emph{core function} of $f$ and is denoted by $P_{C}$. 
\end{definition}

\begin{remark}
Note that in Definition~\ref{Defn_k_can} we require that the function $g$ be unique when all the canalizing variables are not in their corresponding canalizing input sets.
As a result, a function could be canalizing but not $1$-canalizing, see Example~\ref{ex:can_not_1can}.
\end{remark}

\begin{example}\label{ex:can_not_1can}
Let $\mathbb{F}=\{0,1,2\}$ and $n=2$. Consider the function
$$f(x_1,x_2)=1+2 x_1^2+2 x_2+2 x_1^2 x_2+2 x_2^2.$$
For this function $x_2$ is canalizing (with $S_1=\{2\}$) because $f(x_1,2)=1$. However, $f$ is not a $1$-canalizing function because $f(x_1,0)=1+2 x_1^2\neq 2+x_1^2=f(x_1,1)$. Thus, even though $x_2$ is canalizing for $f$ , the function $f$ has no layers of canalization. Thus, $P_C=f$.
\end{example}

\subsection{Layers of canalization in multistate networks}
In Theorem~\ref{thm:layers_multi} we provide a polynomial normal description of discrete functions. Basically, this theorem gives a partition of the inputs of the function into canalizing and non-canalizing variables and, within the canalizing ones, we categorize the input variables into layers of canalization.
This theorem is a generalization of a theorem in~\cite{he2016stratification} from Boolean to the multistate case.

Let $S\subset\mathbb{F}$ be a subset of $\mathbb{F}$ 
and $\widetilde{Q}_{S}(u)$ be the indicator function of the complement of $S$. That is,
\begin{displaymath}
\widetilde{Q}_{S}(x)=\biggl\{\begin{array}{ll}
1&\text{if}\ x\notin S,\\
0&\text{if}\ x\in S.
\end{array}
\end{displaymath}

\begin{theorem}\label{thm:layers_multi}
Every multistate function can be uniquely written as
\begin{equation}
\label{eq:layers_multistate}
\begin{array}{l}
f(x_1,\dots,x_n) =
M_1(M_2(\dots(M_{r-1}(M_rP_C+\\\\
\quad\quad\quad B_r)+B_{r-1})\dots)+B_2)+B_1,\\
\end{array}
\end{equation}
where $M_i = \displaystyle\prod_{j=1}^{k_i}\widetilde{Q}_{S_{i,j}}$, $d = k_1+\cdots+ k_r$ is the canalizing depth,
$P_C$ is a polynomial that has no canalizing variables,
$B_1,B_2,\dots,B_{r}\in\mathbb{F}$, and $B_{r}\neq0$.
Each variable $x_i$ appears in exactly one of the $M_1,M_2,\dots,M_r,P_C$.
\end{theorem}
\begin{proof}
If $f(x_1,\dots,x_n)$ is non-canalizing, then $P_C=f$.
If $f(x_1,\dots,x_n)$ is canalizing, then we proceed by induction.
For $n=1$, if $f$ is canalizing in $x_i$ but not $1$-canalizing in $x_i$, then we set $P_C=f$. If $f$ is $1$-canalizing in $x_i$, then it can be written as $f=\widetilde{Q}_{S_{1}}(x_i)+B_1$ for some set $S_1\subset\mathbb{F}$. Then $f$ has the form of Equation~\ref{eq:layers_multistate} by setting $M_1=\widetilde{Q}_{S_{1}}(x_i)$ and $P_C=1$. For $n=2$, if $f(x_i,x_j)$ is not $1$-canalizing on any of its variables, then we set $P_C=f$. If $f$ is $1$-canalizing on $x_i$, then $f$ can be written as $f(x_i,x_j)=M_1(x_i)g(x_j)+B_1$ for some $g(x_j)$. Then $f$ has the form of Equation~\ref{eq:layers_multistate} by setting $P_C=g$. Now assume that Equation~\ref{eq:layers_multistate} is true for any canalizing function that is essential in at most $n-1$ variables (that is, for all functions that depend in at most $n-1$ variables). Let $f$ be a function that is essential in $n$ variables. If $f$ is not $1$-canalizing on any of its variables, then we set $P_C=f$. If $f$ is $1$-canalizing in $x_{i_1},\dots,x_{i_{k_1}}$, then $f=M_1g+B_1$, where $M_1$ is the product of indicator functions of the complements of sets $S_{i_1},\dots,S_{i_{k_1}}\subset\mathbb{F}$ and $g$ has $n-k_1$ variables. If $g$ has no canalizing variables, then $f$ has the form of Equation~\ref{eq:layers_multistate} with $P_C=g$. If $g$ is canalizing, then by the inductive hypothesis $g$ can be written as
\begin{displaymath}
g = M_2(\dots(M_{r-1}(M_rP_C+B_r)+B_{r-1})\dots)+B_2.
\end{displaymath}
Thus, $f$ has the form of Equation~\ref{eq:layers_multistate}.
\end{proof}
\begin{remark}
\item For a multistate nested canalizing function, the formula in Equation~\ref{eq:layers_multistate} reduces to
\begin{equation}
\label{eq:layers_NCF_multistate}
\begin{array}{l}
f(x_1,\dots,x_n) = M_1(M_2(\dots(M_{r-1}(B_{r+1}M_r+\\
\quad B_r)+B_{r-1})\dots)+B_2)+B_1,
\end{array}
\end{equation}
as was shown in~\cite{kadelka2017multistate}.
\end{remark}

In the following example we describe a $2$-canalizing function with noncanalizing variables.

\begin{example}
Let $\mathbb{F}=\{0,1,2\}$ and $n=4$. Consider the function 
$$
\begin{array}{l}
f(x_1,x_2,x_3,x_4)=1+x_1^2+x_1^2 x_2+2 x_1^2 x_2^2+\\
x_1^2 x_2 x_3+2 x_1^2 x_2^2 x_3+x_1^2 x_2 x_4+2 x_1^2 x_2^2 x_4.
\end{array}
$$ 
The function $f$ can be written as in Equation~\ref{eq:layers_multistate} as  
$$f=M_1(M_2(P_C+1)+1)+1),$$ 
where $M_1=\widetilde{Q}_{S_1}(x_1)=x_1^2$, $S_1=\{0\}$, $M_2=\widetilde{Q}_{S_2}(x_2)=x_2+2 x_2^2$, $S_2=\{0,1\}$, and $P_C=x_3+x_4$. Thus $f$ has two layers and two noncanalizing variables.
Note that $f$ can also be written as in Equation~\ref{eq:defn_k_can} as
\begin{displaymath}
\begin{array}{l}
f(x_1,x_2,x_3,x_4)=\\\\
\left\{
\begin{array}{l}
1,\text{ if}\ x_{1}\in S_1=\{0\},\\
2,\text{ if}\ x_{1}\notin S_1,x_{2}\in S_2=\{0,1\},\\
P_C,\text{ if}\ x_{1}\notin S_1,x_{2}\notin S_2.
\end{array}\right.
\end{array}
\end{displaymath}
\end{example}
\begin{figure}
    \centering
    \includegraphics[width=.5\textwidth]{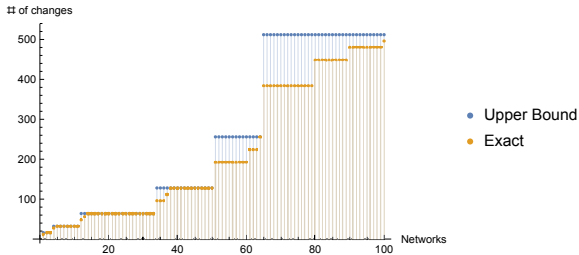}
    \caption{Statistics for the number of changes in the first layer of scale-free Boolean networks. The $x$-axis shows the 100 networks that were randomly generated and the $y$-axis shows the number of changes corresponding to a network in the $x$-axis. In Figure~\ref{fig:errors_layer1_Boolean} we plot the differences between upper bounds and the exact number of changes for these networks.}
    \label{fig:layer1_Boolean}
\end{figure}
\subsection{Upper bounds}
Using the polynomial normal form of multistate functions in Theorem~\ref{thm:layers_multi}, we derive a set of formulas for counting the maximum number of transitions that will change in the state space upon an edge deletion in the wiring diagram. The formulas presented here are generalizations from the Boolean case to the multistate setting of the formulas we presented in~\cite{Murrugarra:2015uq}.

For the next theorem, we are going to assume that the functions of the discrete network $\mathbf{F}=(f_1,\ldots,f_n):\mathbb{F}^n\to \mathbb{F}^n$ are written in the format of Theorem~\ref{thm:layers_multi}. That is, for $t=1,\dots,n$ the coordinate function $f_t$ has the following form,
\begin{equation}\label{eq:layers_del_bound}
\begin{array}{l}
f_t(x_1,\dots,x_n) = M_1^t(M_2^t(\dots(M_{r-1}^t(M_r^tP_C+\\\\
\quad\quad\quad\quad B_r)+B_{r-1})\dots)+B_2)+B_1,
\end{array}
\end{equation}
where $M^t_i = \displaystyle\prod_{j=1}^{k_i}\widetilde{Q}_{S_{j,t}}$, $d = k_1+\cdots+ k_r$ is the canalizing depth,
$P_C$ is a polynomial with no canalizing variables,
$B_1,B_2,\dots,B_{r}\in\mathbb{F}$, and $B_{r}\neq0$.
Each variable $x_i$ appears in exactly one of $M_1^t,M_2^t,\dots,M_r^t,P_C$.
\begin{remark}
Note that the function $f_t$ has $r$ layers and there are $k_i$ variables in each layer for $i=1,\dots,r$. 
\end{remark}
In the following theorem, we assume that the canalizing input sets are all of the same size for all the variables. In Theorem~\ref{thm:multistate_bound} we study the general case where the canalizing input sets of the variables can be different.
\begin{theorem}\label{thm:same_input_set}
Let $\mathbf{F}=(f_1,\ldots,f_n):\mathbb{F}^n\to \mathbb{F}^n$ be a multistate network where $f_t$ is a $k$-canalizing function written as in Eq.~\ref{eq:layers_del_bound} with $k_1,\ldots, k_r$ the numbers of variables in layers $1,\ldots, r$, respectively. Let $x_s$ be in the $\ell^{th}$ layer, where $\ell\le r$ and $r$ is the number of layers. Suppose that all canalizing variables have the same canalizing input set $S$ and that $0\in S$. Then, the maximum number of transitions in the state space that will change upon deletion of $x_s\to x_t$ is given by
\begin{equation}\label{eq:bound_same_inputs}
\begin{array}{l}
p^{n-k_1-\cdots-k_\ell}\left( p-\left| S\right| \right)^{k_1+\cdots+k_\ell}.
\end{array}
\end{equation}
\end{theorem}
\begin{proof}
Let $m=k_1+\cdots+k_\ell$. The number of input vectors where the other  canalizing variables (not $x_s$) of $f_t$ do not take on their canalizing input is $\left( p-\left| S\right| \right)^{m-1}$. For these input vectors, if $x_s$ was already set to 0 or to any other of its canalizing values in $S$, then the output of $f_t$ will not change as a result of deleting $x_s\to x_t$. Finally, since we have $n-m$ non-canalizing variables, the total number of input vectors for which the output of $f_t$ can possibly change is $\left( p-\left| S\right| \right)^{m-1}(p-\left| S\right|)p^{n-m}=\left( p-\left| S\right| \right)^{m}p^{n-m}$.
\end{proof}
\begin{remark}
Note that from Equation~\ref{eq:bound_same_inputs} that the number of variables in each layer affects the number of changes and that there are potentially more changes when the deletion happens in a more dominant layer, see examples~\ref{ex:Boolean_case}-\ref{ex:Multistate_case}.
\end{remark}
\begin{theorem}\label{thm:multistate_bound}
Let $\mathbf{F}=(f_1,\ldots,f_n):\mathbb{F}^n\to \mathbb{F}^n$ be a multistate network where $f_t$ is a $k$-canalizing function written as in Eq.~\ref{eq:layers_del_bound} with $k_1,\ldots, k_r$ the numbers of variables in layers $1,\ldots, r$, respectively.
Let $x_s$ be in the $\ell^{th}$ layer, $\ell\le r$ and $r$ is the number of layers. The maximum number of transitions in the state space that will change upon deletion of $x_s\to x_t$ is given by
\begin{equation}\label{eq:formula_bound}
\begin{array}{l}
p^{n-k_1-\cdots-k_\ell}\cdot\left(\displaystyle\prod_{i=1}^{\ell-1}~~\prod_{j=1}^{k_i}(p-|S_{j,t}|) \right)\\
\left(\displaystyle\prod_{\begin{subarray}{c} j=1\\ j\neq s \end{subarray}}^{k_\ell}(p-|S_{j,t}|)\right)(p-R),
\end{array}
\end{equation}
where 
\begin{displaymath}
R = \left\{ \begin{array}{cc}
 |S_{s,t}| & \textrm{if $0\in S_{s,t}$}\\
 p-|S_{s,t}| & \textrm{if $0\notin S_{s,t}$.}
  \end{array} \right.
\end{displaymath}
\end{theorem}
\begin{proof}
The strategy is to first count the number of inputs that do not contain values from the canalizing sets of the variables in the first $\ell-1$ layers (that do not contain $x_s$). Thus, the term in the first line of Equation~\ref{eq:formula_bound_proof} counts the number of non-canalizing inputs in the previous layers to the layer containing $x_s$; the term inside the second set of parentheses of Equation~\ref{eq:formula_bound_proof} counts the number of non-canalizing inputs of the variables (except of $x_s$) in the layer containing $x_s$; the last term in Equation~\ref{eq:formula_bound_proof} counts the number of non-canalizing inputs of $x_s$. For the last term, notice that deleting $x_s\to x_t$ results in setting $x_s=0$ in $f_t$. If 0 is in the canalizing set of $x_s$, $S_{s,t}$, then the rest of the values in $S_{s,t}$ will yield the same output as 0. Since $|S_{s,t}|/p$ of the input values in the transition table of $f_t$ contain a canalizing value for $x_s$, it is the remaining $\frac{p-|S_{s,t}|}{p}$ of the table that can potentially change as a result of the edge deletion. On the other hand, if $0\notin S_{s,t}$, then it is the inputs not in $S_{s,t}$ that have the potential to change the output as a result of deleting $x_s\to x_t$ which constitutes $1/p$ of the transition table, with $\frac{p-1}{p}$ of the table that can potentially change as a result of the edge deletion. Thus, to obtain Equation~\ref{eq:formula_bound} we multiply the following expressions:
\begin{equation}\label{eq:formula_bound_proof}
\begin{array}{l}
\frac{p^{n}}{p^{k_1+\cdots+k_{\ell-1}}}\cdot\left(\displaystyle\prod_{i=1}^{\ell-1}~~\prod_{j=1}^{k_i}(p-|S_{j,t}|) \right)\\
\left(\frac{1}{p^{k_\ell-1}}\displaystyle\prod_{\begin{subarray}{c} j=1\\ j\neq s \end{subarray}}^{k_\ell}(p-|S_{j,t}|)\right)\frac{1}{p}(p-R)=\\
p^{n-k_1-\cdots-k_\ell}\cdot\left(\displaystyle\prod_{i=1}^{\ell-1}~~\prod_{j=1}^{k_i}(p-|S_{j,t}|) \right)\\
\left(\displaystyle\prod_{\begin{subarray}{c} j=1\\ j\neq s \end{subarray}}^{k_\ell}(p-|S_{j,t}|)\right)(p-R),
\end{array}
\end{equation}
\end{proof}
\begin{figure}
    \centering
    \includegraphics[width=.5\textwidth]{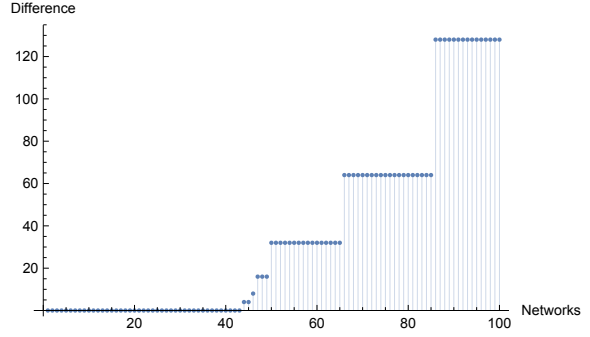}
    \caption{Statistics for the differences between the upper bounds and the exact number of changes for the networks in Figure~\ref{fig:layer1_Boolean}. In about 40\% of the networks the upper bounds match the exact number of changes.}
    \label{fig:errors_layer1_Boolean}
\end{figure}
\begin{remark}
\begin{enumerate}
\item The bound in Equation~\ref{eq:formula_bound} is sharp.
\item When $p=2$, the formula in Equation~\ref{eq:formula_bound} reduces to $2^{n-k_1-k_2-\cdots-k_{r}}$.
\item If instead of edge deletion, we consider constant expression to $a\in\mathbb{F}$ (see Section~\ref{subsec:ce}) of $x_s\to x_t$, then the formula in Equation~\ref{eq:formula_bound} remains the same except for $R$ which becomes
\begin{displaymath}
R = \left\{ \begin{array}{cc}
 |S_{s,t}| & \textrm{if $a\in S_{s,t}$}\\
 p-|S_{s,t}| & \textrm{if $a\notin S_{s,t}$.}
  \end{array} \right.
\end{displaymath}
\end{enumerate}
\end{remark}
\begin{figure}
    \centering
    \includegraphics[width=.5\textwidth]{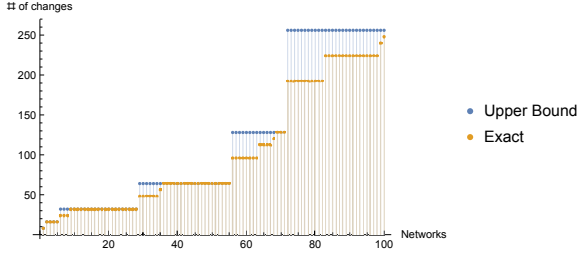}
    \caption{Statistics for the number of changes in the second layer of scale-free Boolean networks. The $x$-axis shows the 100 networks that were randomly generated and the $y$-axis shows the number of changes corresponding to a network in the $x$-axis. In Figure~\ref{fig:errors_layer2_Boolean} we plot the differences between upper bounds and the exact number of changes for these networks.}
    \label{fig:layer2_Boolean}
\end{figure}
\begin{proposition}\label{thm:multistate_not_canal_bound}
Let $\mathbf{F}=(f_1,\ldots,f_n):\mathbb{F}^n\to \mathbb{F}^n$ be a multistate network where $f_t$ is written as in Equation~\ref{eq:layers_del_bound}.
Let $x_s\in supp(P_C)$. The maximum number of transitions in the state space that will change upon deletion of $x_s\to x_t$ is 
\begin{equation}\label{eq:bound_p_C}
p^{n-k_1-\cdots-k_{r}-1}\left(\prod_{i=1}^r~\prod_{j=1}^{k_i}(p-|S_{j,t}|)\right)(p-1).
\end{equation}
\end{proposition}
\begin{figure}
    \centering
    \includegraphics[width=.5\textwidth]{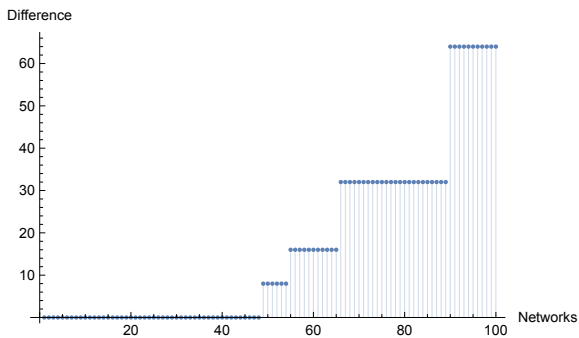}
    \caption{Statistics for the differences between the upper bounds and the exact number of changes for the networks in Figure~\ref{fig:layer2_Boolean}. In about 50\% of the networks the upper bounds match the exact number of changes.}
    \label{fig:errors_layer2_Boolean}
\end{figure}
\begin{remark}
\begin{enumerate}
\item This upper bound is sharp.
\item When $p=2$, the expression reduces to $2^{n-d-1}$, where $d$ is the canalizing depth.
\item If $f$ has no canalizing variables, then the formula in Equation~\ref{eq:bound_p_C} reduces to $p^{n-1}(p-1)$.
\end{enumerate}
\end{remark}

\begin{proposition}
Let $\mathbf{F}=(f_1,\ldots,f_n):\mathbb{F}^n\to \mathbb{F}^n$ be a multistate network where $f_t$ is written as in Equation~\ref{eq:layers_del_bound}.
If $P_C$ is canalizing but not $1$-canalizing with canalizing variable $x_s$ and input set $S_{s,t}$, then 
there are two cases to consider:
\begin{enumerate}
    \item The deletion of $x_s\to x_t$ will result in up to
\begin{equation}\label{eq:bound_can_but_not_1can}
p^{n-k_1-\cdots-k_{r}-1}\left(\prod_{i=1}^r~\prod_{j=1}^{k_i}(p-|S_{j,t}|)\right)(p-R)
\end{equation}
transitions, where 
\begin{displaymath}
R = \left\{ \begin{array}{cc}
 |S_{s,t}| & \textrm{if $0\in S_{s,t}$}\\
 1 & \textrm{if $0\notin S_{s,t}$.}
  \end{array} \right.
\end{displaymath}
    \item Let $x_a\in supp(P_C)$ that is not canalizing. Then 
the maximum number of transitions in the state space that will change upon deletion of $x_a\to x_t$ is 
    \begin{equation}\label{eq:bound_can_but_not_1can1}
    \begin{array}{l}
p^{n-k_1-\cdots-k_{r}-2}\left(\displaystyle\prod_{i=1}^r~\prod_{j=1}^{k_i}(p-|S_{j,t}|)\right)\\
(p-|S_{s,t}|)(p-1).
\end{array}
\end{equation}    
\end{enumerate}
\end{proposition}
\begin{remark}
\begin{enumerate}
\item This upper bound is sharp.
\item If $P_C$ has more than one canalizing variables (but it is still not 1-canalizing), then the formula in Equation~\ref{eq:bound_can_but_not_1can1} becomes
    \begin{equation}\label{eq:bound_can_but_not_1can2}
    \begin{array}{l}
p^{n-k_1-\cdots-k_{r}-c-1}\left(\displaystyle\prod_{i=1}^r~\prod_{j=1}^{k_i}(p-|S_{j,t}|)\right)\\
\prod_{i=1}^c(p-|S_{s_i,t}|)(p-1),
\end{array}
\end{equation}
where each $x_{s_i}$ is a canalizing variable and $c$ is the number of canalizing variables of $P_C$.
\end{enumerate}
\end{remark}
\begin{figure}
    \centering
    \includegraphics[width=.5\textwidth]{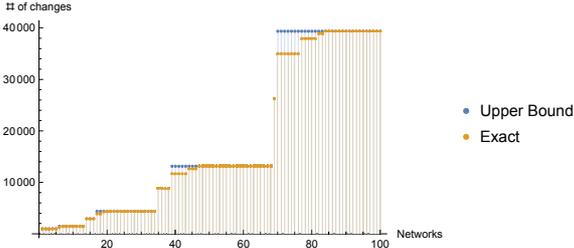}
    \caption{Statistics for the number of changes in the first layer of scale-free multistate networks. The $x$-axis shows the 100 networks that were randomly generated and the $y$-axis shows the number of changes corresponding to a network in the $x$-axis. In Figure~\ref{fig:errors_layer1_multistate} we plot the differences between upper bounds and the exact number of changes for these networks.}
    \label{fig:layer1_multistate}
\end{figure}
\section{Applications}\label{sec:Applications}
To provide further insights into the results presented above and to illustrate the use of the formulas here we
present numerical results for random networks where we compare the exact number of changes to the upper bounds provided by the formulas.

For the next examples, we generated random networks with scale-free structure using the Barabasi-Albert algorithm~\cite{Barabasi:1999aa}.
We note that the Barabasi-Albert algorithm produces undirected edges but for our examples we need directed edges. Thus, for each undirected edge
between $x_i$ and $x_j$, we converted the edge $x_i - x_j$ into either $x_i\rightarrow x_j$ or $x_j\rightarrow x_i$ at random.

\begin{example}[Boolean Case]\label{ex:Boolean_case}
In order to calculate the exact number of changes we use random networks with 10 nodes.
For each network, we selected the node with the maximum in-degree and generated a random partition of its inputs to assign the canalizing layers. The functions of the other nodes were generated at random.

In Figure~\ref{fig:layer1_Boolean} we show statistics for the number of changes in the first layer.
The average maximum in-degree for the networks in Figure~\ref{fig:layer1_Boolean} is $4.14$ ($std=1.07$). The average number of variables in the first layer is $2.61$ ($std=1.5$). The average number of changes in the first layer is $221.28$ ($std=168.363$) and the average upper bound is $259.04$ ($std=201.562$).

In Figure~\ref{fig:errors_layer1_Boolean} we present statistics of the number of changes as well the difference between the exact number of changes and the upper bound provided by the formulas. For these networks, in about 40\% of the cases the upper bounds match the exact number of changes.

In Figure~\ref{fig:layer2_Boolean} we show statistics for the number of changes in the second layer.
The average maximum in-degree for the networks in Figure~\ref{fig:layer2_Boolean} is $4.8$ ($std=1.07$). 
The average number of variables in the first layer is $1.67$ ($std=0.93$). 
The average number of variables in the second layer is $2.13$ ($std=0.75$).
The average number of changes in the second layer is $86.0$ ($std=62.7$) and the average upper bound is $100.64$ ($std=78.88$).

In Figure~\ref{fig:errors_layer2_Boolean} we present statistics of the number of changes as well the difference between the exact number of changes and the upper bound provided by the formulas. For these networks, in about 50\% of the cases the upper bounds match the exact number of changes.

From Figures~\ref{fig:errors_layer1_Boolean} and \ref{fig:errors_layer2_Boolean}, we see that the bounds are slightly more accurate when the edge intervention happens in a less dominant layer. 
\end{example}
\begin{figure}
    \centering
    \includegraphics[width=.5\textwidth]{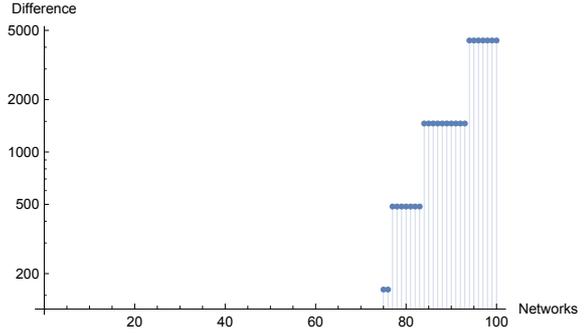}
    \caption{Statistics for the differences between the upper bounds and the exact number of changes for the networks in Figure~\ref{fig:layer1_multistate}. In about 75\% of the networks the upper bounds match the exact number of changes. The vertical axis is in logarithmic scale.}
    \label{fig:errors_layer1_multistate}
\end{figure}
\begin{example}[Multistate Case]\label{ex:Multistate_case}
Here we also use random networks with scale-free structure with $p=3$ and $n=10$ nodes.
For each network, we selected the node with the maximum in-degree and generated a random partition of its inputs to assign the canalizing layers. The functions of the other nodes were generated at random.

In Figure~\ref{fig:layer1_multistate} we show statistics for the number of changes in the first layer.
The average maximum in-degree for the networks in Figure~\ref{fig:layer1_multistate} is $4.05$ ($std=0.88$). The average number of variables in the first layer is $2.28$ ($std=1.16$). The average number of changes in the first layer is $17303.2$ ($std=14739.4$) and the average upper bound is $17792.5$ ($std=15220.6$).

In Figure~\ref{fig:errors_layer1_multistate} we present statistics of the difference between the upper bounds provided by the formulas and the exact number of changes. For these networks, in about 75\% of the cases the upper bounds match the exact number of changes.

In Figure~\ref{fig:layer2_multistate} we show statistics for the number of changes in the second layer.
The average maximum in-degree for the networks in Figure~\ref{fig:layer2_multistate} is $4.84$ ($std=0.94$). 
The average number of variables in the first layer is $1.7$ ($std=0.86$). 
The average number of variables in the second layer is $1.93$ ($std=0.97$).
The average number of changes in the second layer is $3946.68$ ($std=3789.71$) and the average upper bound is $4056.48$ ($std=3969.7$).

\begin{figure}
    \centering
    \includegraphics[width=.5\textwidth]{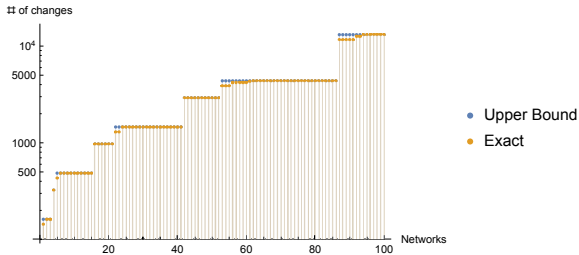}
    \caption{Statistics for the number of changes in the second layer of scale-free networks with $p=3$ and $n=10$ nodes. The $x$-axis shows the 100 networks that were randomly generated and the $y$-axis shows the number of changes corresponding to a network in the $x$-axis. The vertical axis is in logarithmic scale. In Figure~\ref{fig:errors_layer2_multistate} we plot the differences between upper bounds and the exact number of changes for these networks.}
    \label{fig:layer2_multistate}
\end{figure}

In Figure~\ref{fig:errors_layer2_multistate} we present statistics of the difference between the the upper bound provided by the formulas and the exact number of changes. For these networks, in about 80\% of the cases the upper bounds match the exact number of changes.

\begin{figure}
    \centering
    \includegraphics[width=.5\textwidth]{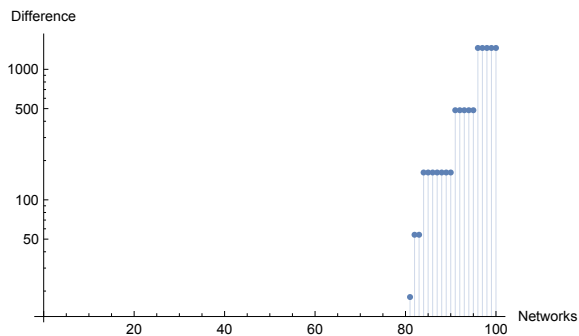}
    \caption{Statistics for the differences between the upper bounds and the exact number of changes for the networks in Figure~\ref{fig:layer2_multistate}. In about 80\% of the networks the upper bounds match the exact number of changes. The vertical axis is in logarithmic scale.}
    \label{fig:errors_layer2_multistate}
\end{figure}

From Figures~\ref{fig:errors_layer1_multistate} and \ref{fig:errors_layer2_multistate}, we see that the bounds are slightly more accurate when the edge intervention happens in a less dominant layer. 
\end{example}
\section{Conclusions}\label{sec:Conlusions}
In this paper we present practical methods for quantifying the global changes that result from an application of an external intervention in the network, which we called the total effect of the intervention. 
We emphasized that, while there are several methods for identifying control targets in discrete networks, there have been very few studies focusing on the total extent of the consequences of applying a certain control action
(beyond satisfying the control objective). This paper contributes methods for measuring the number of changed transitions in the state space upon the application of an edge control in multistate networks. The approach is based on a polynomial normal form description of discrete functions that provides a way of categorizing the inputs of the function and therefore of quantifying their impact on the dynamics of the network.
The main computational cost of our approach is in obtaining the canalizing layers format of the functions which we used to derive our formulas. Once the functions are written in the layers format, it is straightforward to apply the formulas to calculate the upper bound. We note that obtaining the layers format could be a formidable task with exponential complexity on the number of inputs in the worst case but for the type of networks we are studying (i.e. biological networks) we believe that our approach is still practical. We applied our methods to randomly generated multistate models and verified that in many cases the upper bounds provided by our formulas were accurate. We also observed that the upper bounds tend to be more accurate when the edge interventions happen in a less dominant layer.
\textbf{\bibliographystyle{agsm}}
\bibliography{ref_for_control}
\end{document}